\documentclass[11pt, twoside, letterpaper]{amsart}
\usepackage{amsmath, hyperref, color, amssymb}
\usepackage{enumerate}   
\usepackage{srcltx}
\usepackage{soul}
\usepackage{cancel}
\usepackage[normalem]{ulem}
 
\makeatletter
\@namedef{subjclassname@2020}{%
  \textup{2020} Mathematics Subject Classification}
\makeatother

\theoremstyle{plain}%
\newtheorem{Theorem}{Theorem}[section] %
\newtheorem{Lemma}[Theorem]{Lemma}
\newtheorem{Proposition}[Theorem]{Proposition} %

\theoremstyle{definition}%
\newtheorem{Assumption}[Theorem]{Assumption}%

\theoremstyle{remark}%
\newtheorem{Remark}[Theorem]{Remark} %
\newcommand{\cA}{\mathcal{A}}

\newcommand{\cD}{\mathcal{D}}
\newcommand{\cE}{\mathcal{E}}
\newcommand{\cF}{\mathcal{F}}

\newcommand{\cM}{\mathcal{M}}

\newcommand{\cS}{\mathcal{S}}

\newcommand{\bE}{\mathbb{E}}

\newcommand{\bL}{\mathbb{L}}

\newcommand{\bN}{\mathbb{N}}

\newcommand{\bP}{\mathbb{P}}

\newcommand{\bR}{\mathbb{R}}

\newcommand{\bY}{\mathbb{Y}}
\newcommand{\bZ}{\mathbb{Z}}

\newcommand{\e}{\varepsilon}

\newcommand{\Pas}{{\mathbb{P}\text{-a.s.}}}

\newcommand{\dbracc}[1]{[\kern-0.15em[ #1 ]\kern-0.15em]}
\newcommand{\dbraco}[1]{[\kern-0.15em[ #1 [\kern-0.15em[}
\newcommand{\dbraoc}[1]{]\kern-0.15em] #1 ]\kern-0.15em]}
\newcommand{\dbraoo}[1]{]\kern-0.15em] #1 [\kern-0.15em[}
\newcommand{\be}{\begin{equation}}
\newcommand{\ee}{\end{equation}}
\newcommand{\nn}{\nonumber}
\newcommand{\bs}{\begin{split}}
\newcommand{\es}{\end{split}}
\newcommand{\ba}{\begin{aligned}}
\newcommand{\ea}{\end{aligned}}

\renewcommand{\[}{\left[}
\renewcommand{\]}{\right]}
\renewcommand{\(}{\left(}
\renewcommand{\)}{\right)}

\newcommand{\D  }{\Delta}
\renewcommand{\d }{\Delta}

\DeclareMathOperator{\esssup}{ess\,sup}
\DeclareMathOperator{\essinf}{ess\,inf}

\newcommand{\tr}[1]{\textcolor{black}{#1}}

\numberwithin{equation}{section}

\renewcommand{\[}{\left[}
\renewcommand{\]}{\right]}
\renewcommand{\(}{\left(}
\renewcommand{\)}{\right)}

\begin{document}

\title[Stability of the Epstein-Zin Problem]
{Stability of the Epstein-Zin Problem}
\author{Michael Monoyios}
\address{Michael Monoyios, Mathematical Institute, University of Oxford, Radcliffe Observatory Quarter, Woodstock
Road, Oxford OX2 6GG, UK}
\email{monoyios@maths.ox.ac.uk}
 \author{Oleksii Mostovyi}
 \address{
Oleksii Mostovyi, Department of Mathematics, University of
Connecticut, Storrs, CT 06269, United States}
\email{oleksii.mostovyi@uconn.edu}
\thanks{Oleksii Mostovyi is the corresponding author. He has been supported by the National Science Foundation under grant No. DMS1848339. Part of this work was completed while he was visiting the University of Oxford. } 
\begin{abstract}
We investigate the stability of the Epstein-Zin problem with respect to small distortions in the dynamics of the traded securities. We work in incomplete market model settings, where our parametrization of perturbations allows for joint distortions in returns and volatility of the risky assets and the interest rate. 
Considering empirically the most relevant specifications of risk aversion and elasticity of intertemporal substitution, we provide a condition that guarantees the convexity of the domain of the underlying problem and results in the existence and uniqueness of a solution to it. Then, we prove the convergence of the optimal consumption streams, {the associated wealth processes, the indirect utility processes}, and the value functions in the limit when the model perturbations vanish. 
\end{abstract}
\date{\today}%

\maketitle
\section{Introduction}
Recursive utilities of Epstein-Zin type allow for the incorporation of {\it future} consumption choice into preferences. In the discrete-time environment, this topic goes back to \cite{KrepsPort} and \cite{EpsteinZin}, whereas in continuous-time stochastic settings, 
it was originally investigated in \cite{DuffieEpstein92b}.  These utilities allowed for the resolution of several asset pricing puzzles; see the introduction to \cite{Hao17}  for an overview of this topic. The Epstein-Zin problem remains an active research area.
  Thus, recently, explicit solutions are characterized in \cite{Hao17}, \cite{Kraft17}, and \cite{MX};   for the results in infinite time horizon settings, we refer to  \cite{HerdHob2}, \cite{HerdHob1}, and \cite{huang2}; a finite yet random horizon is considered in \cite{huang1}.

The continuous-time counterpart of a recursive utility is also known as a stochastic differential utility. 
Two constants govern its parametrization. 
One is  the usual risk aversion, and the other is an elasticity of intertemporal substitution (EIS) that specifies the willingness to interchange consumption over time. 
%
As pointed out in \cite[Remark 2.1, p. 231]{Hao17}, the empirically most relevant case corresponds to the relative risk aversion $\gamma>1$ and the elasticity of intertemporal substitution (EIS) $\psi>1$. 
 
 {
 The notion of the well-posedness of a mathematical problem goes back to Hadamard \cite{Had02}, and it  comprises the following three properties for a solution to a given problem to hold:
existence, uniqueness, and continuous dependence on the initial data, where the last property is loosely known as stability. While the existence and uniqueness results (and various characterizations of the solution) for the Epstein-Zin problem are established in the papers mentioned above, the questions of stability in the context of this problem, to the best of our knowledge, have not been answered before. 
 }
 
 {
 An additional motivation for studying stability comes from the fact that in many cases, for example, in the factor model considered in \cite{Kraft17}, the explicit solution ceases to exist under general perturbations of the model parameters, where such perturbations can be associated with a procedure of calibration. In this case, it is important to understand whether the outputs of the problem, such as the optimal consumption, the optimal wealth process, the indirect utility process, and the value function, differ only slightly from the solution to the unperturbed problem admitting an explicit solution.  
 }
 
 {
 In the case of the more traditional additive utility, which corresponds to a particular case of the Epstein-Zin problem ($\gamma = \frac 1\psi$, in the present notations), the questions of stability are studied more, and historically, and they have also followed establishing the existence and uniqueness results. The results on the stability of the outputs to the optimal investment problem with respect to various perturbations and in varying formulations are contained in 
 \cite{JouiniNapp}, \cite{CarRas}, \cite{KarZit11}, \cite{HaoStab17}, \cite{VeraguasSilva}, and \cite{MostovyiFPP}, among others.
 These works do not establish any stability to BSDEs result, in contrast to the present paper, as the analysis of the stability of the optimal investment problem in many formulations relies on different techniques, despite the BSDE-base approach pioneered in \cite{imHuMuller05}.  Thus, compared to the papers on the stability of the traditional utility maximization in various formulations mentioned in this paragraph, we rely on the analysis of BSDE and establish related approximation and stability results in the present work. 
 }

{In view of the previously listed works, one can argue that the literature on the Epstein-Zin problem does not contain its stability analysis. The aim of the present paper is to
give insight into this problem, and thus,} here, we investigate the stability of the Epstein-Zin problem with respect to perturbations of the dynamics of the traded securities. Our parametrization of perturbations allows us to include joint or separate distortions of the interest rate as well as of the return and volatility of the risky assets. We consider the above-described case when 
both the relative risk aversion and EIS exceed one. 
 Our analysis is performed under a weak no-arbitrage condition,
no unbounded profit with bounded risk (NUPBR) introduced in \cite{KarKar07}, which still allows
for the meaningful structure of the underlying problem. 

Our results include a sufficient condition for the convexity of the domain of the primal problem and for the existence and uniqueness of the optimizer to this problem. This condition can be stated as non-emptiness of the dual domain, that is, the existence of a state price density satisfying an integrability condition, which guarantees a unique solution of class D to the dual BSDE, see {Lemma \ref{lemConvA}}. 
We also show the convergence of the value functions, the optimal consumption streams, {the associated wealth processes, and the indirect utility processes} as perturbations vanish. 

One of the difficulties in the analysis involves establishing stability-type estimates for the solutions to BSDEs with an unbounded terminal condition and non-Lipschitz generator, with respect to particular perturbations of both the terminal condition and the generator. Here, we establishe a ucp convergence result for the family of solutions to such BSDEs, see Lemma \ref{lemck}.  
Further, it is crucial for the proof to show the strict {(and stronger than strict)} concavity of the {value function}, in a sense Lemma \ref{lemConv}. All these estimates are needed to establish the convergence of the optimal consumption streams, whereas the convergence of the value functions relies on conjugacy results from \cite{MX}
and on a particular construction of the nearly optimal consumption streams also combined with localization.

The remainder of this paper is organized as follows: in Section \ref{secModel}, we specify the model; Section \ref{secMainResults} contains the main results. In Section \ref{secInteg}, we discuss the integrability condition on the perturbations, and the proofs are given in Section \ref{secProofs}. 
\section{Model}\label{secModel}
\subsection{Market}

Let $\(\Omega, \(\cF_t\)_{t\in[0,T]}, \cF, \bP\)$ be a complete stochastic basis, where $T\in(0,\infty)$ is the time horizon,  {$\cF_0$ is the completion of the trivial $\sigma$-field}, $\(\cF_t\)_{t\in[0,T]}$ is the augmented filtration generated by a $(k+n)$-dimensional Brownian motion  $B = \(W, W^\perp\)$, where $W$ represents the first $k$ components and $W^\perp$ the remaining $n$ components. For an $\bR^{n\times n}$-valued $\cF^W$-adapted process $\rho$ taking values in $\bR^{n\times k}$ and for an $\bR^{n\times n}$-valued adapted process\footnote{{Through process $\rho$, one, in particular, can include stochastic volatility-type models as in \cite[Section 4]{Kraft17}. We note that the model in \cite{Kraft17} allows for an explicit solution to the Epstein-Zin problem via  Hamilton-Jacobi-Bellman  equations, and under perturbations of the model parameters, the structure allowing for explicit solutions can be lost.}} $\rho^\perp$ satisfying $\rho\rho^\top + \rho^\perp(\rho^\perp)^\top = I_{n\times n}$, the $n$-dimensional identity matrix, we set 
$$W^\rho := \int_0^\cdot \rho_sdW_s  + \int_0^\cdot \rho^\perp_s dW^\perp_s.$$

We consider a family of markets parametrized by $\e\in(-\e_0,\e_0)$ for some $\e_0>0$. Thus, for a fixed $\e$, the traded assets are $(S^{\e, 0}, \dots, S^{\e, n})$, where $S^{\e, 0}$ is the price process of the riskless asset and $(S^{\e, 1},\dots, S^{\e, n})$ are the prices of the risky assets. Their evolution is given by 
\be\label{Se}\bs
dS^{\e, 0}_t = S^{\e, 0}_t r^\e_t dt, \quad 
dS^{\e, i}_t = S^{\e, i}_t\( \(r^\e_t + \mu^{\e, i}_t\)dt + 
\sum\limits_{j = 1}^n \sigma^{\e, i,j}_t dW^{\rho, j}_t\), \\
 i\in\{1, \dots, n\},
\end{split}
\ee
where the processes $r^\e\geq 0$, $\mu^{\e, i}$, and $\sigma^{\e, i}$ are $\cF^W$-adapted processes such that the integrals in \eqref{Se} are well-defined and such that 
$\sigma_t^\e$ is invertible, $t\in[0,T]$, $\Pas$.

In particular, our parametrization of perturbations allows us to include the following cases:

\begin{itemize}
\item Perturbations of the drift $\mu$ only. 
This corresponds to setting $r^\e \equiv r^0$, and $\sigma^{\e, i,j}\equiv \sigma^{0, i,j}$, for every $i, j\in\{1,\dots, n\}$. 
\item Perturbations of the volatility $\sigma$ only. 
\item Similarly, we can consider perturbations of the interest rate only. In many works in mathematical finance, the riskless asset is assumed to be constant-valued. While this gives the correct structure  to many problems of mathematical finance, having non-zero interest rates can also be significant and leads to extra technicalities. 
\item Perturbations of the num\'eraire, where the parametrization of such perturbations can follow the ones in \cite{MostovyiNumeraire}. 
\item Combinations of perturbations above.
\end{itemize}


\subsection{The Epstein-Zin problem}
For every $\e\in(-\e_0, \e_0)$, let \newline $\pi^\e = \(\pi^{\e, 0}, \dots, \pi^{\e, n}\)$ be an $S^\e$-integrable $\bR^{n+1}$-valued process representing the proportions of the total wealth invested in the respective assets, thus,  satisfying {$\sum\limits_{i = 0}^n \pi^{\e, i}_t = 1$, $t\in\[0,T\]$}. 
Let $c^\e$ be a nonnegative {progressively measurable} process representing the consumption rate in the $\e$-th market.  Let $\kappa$ be a deterministic consumption clock given by  $\kappa_t = t + 1_{\{T\}}(t)$, $t\in[0,T]$.  We specify the dynamics of the wealth process $X^{\e, \pi^\e, c^\e}$ associated with consumption-investment pair $(\pi^\e, c^\e)$ and starting from an initial wealth $x$ as follows 
\be\label{Xe}
dX^{\e, \pi^\e, c^\e}_t = X^{\e, \pi^\e, c^\e}_t\sum\limits_{i = 0}^n\pi^{\e, i}_t \frac {dS^{\e, i}_t}{S^{\e, i}_t} - c^\e_t d\kappa_t,\quad X^\e_0 = x.
\ee
We call a consumption process $c^\e$  admissible from $x>0$ for the $\e$-th market, if there exists an $S^\e$-integrable process $\pi^\e$, such that {$\sum\limits_{i = 0}^n \pi^{\e, i}_t = 1$, $t\in\[0,T\]$,} and the associated wealth process in \eqref{Xe} is nonnegative, $\Pas$. We denote the family of {\it admissible} consumptions from  $x>0$ in the $\e$-th market by $\cA(x, \e)$, $\e\in(-\e_0,\e_0)$. 
 
 An agent, starting with an initial capital $x>0$, invests and consumes in the market in a way to maximize his or her expected utility with Epstein-Zin preferences. With $\delta>0$ representing the discount rate, $0<\gamma\neq 1$ being the relative risk aversion, and $0<\psi\neq 1$ specifying the elasticity of inter-temporal substitution (EIS), one can define the 
  Epstein-Zin aggregator $f$ via
\be\label{deff}
f(c,u) = \delta \frac{c^{1 - \frac 1\psi}}{1 - \frac 1\psi}\((1 - \gamma) u\)^{1 - \frac 1 \theta} - \delta \theta u,\quad c>0\quad {\rm and} \quad (1-\gamma)u >0,
\ee
where $\theta := ({1 - \gamma})/\({1-\frac 1 \psi}\)$. Given the bequest utility $U_T(c) =   {c^{1-\gamma}}/\({1-\gamma}\)$, $c>0$, 
the Epstein-Zin utility for a nonnegative consumption stream $c$  is a process $\(U^c_t\)_{t\in[0,T]}$, which satisfies the BSDE
\be\label{defUEZ}
U_t^c = \bE_t\[U_T(c_T) + \int_t^T f(c_s, U^c_s)ds \],\quad t\in[0,T],
\ee
where $\bE_t$ is $\bE\[~\cdot~ |\cF_t\]$. 
{Sufficient conditions for the existence of $U^c$ for a given $c$ are contained in \cite[Proposition 2.1]{MX}.}

The agent aims to maximize his or her Epstein-Zin utility at time zero over all admissible strategies, that is 
\be\label{primalProblemOld}
\sup\limits_{c\in\cA(x,\e)}U^c_0,\quad (x,\e)\in(0,\infty)\times \(-\e_0,\e_0\).
\ee
This formulation, however, does not  guarantee that for a given $c\in\cA(x,\e)$, $U^c$ in \eqref{defUEZ} is well-defined. 
As pointed out in \cite[Remark 2.2]{MX}, one needs some mild integrability properties on the elements of $\cA(x,\e)$, $(x,\e)\in(0,\infty)\times(-\e_0,\e_0)$. 
Below, we provide some insights on this issue. For this, we need to introduce the state price densities. 

 \subsection{{State price} density processes}
 The family of state price density processes is defined as 
 \be\label{defD}\begin{split}
 \cD(y, \e): = &\left\{D>0:\right. D_0 = y, DX^{\e,\pi,c} + \int_0^\cdot D_sc_sd\kappa_s \\
 &\left.~is~a~supermartingale~for~every~ c\in\cA(1, \e)\right\},\\
 &\hspace{35mm} (y,\e)\in(0,\infty)\times \(-\e_0,\e_0\),\\
 \end{split}
 \ee
 where $(\pi, c)$ is the  investment-consumption pair, such that $X^{\e,\pi,c}$ in \eqref{Xe} is nonnegative. 
Thus, one can see that  the family of  minimal state price densities  
\be\label{eqMinStatePriceDensity}
D^{\e, 0}: =  \cE\(-\int_0^\cdot r^\e_sds-\int_0^\cdot \(\(\sigma^{\e}_s\)^{-1} \mu^{\e}_s\)  dW^\rho_s\),\quad \e\in\(-\e_0,\e_0\),
\ee
is well-defined, where $\cE$ denotes the stochastic exponential. 
 In particular, since, for every $\e$, the set of state price densities is non-empty, and this also applies to the set of supermartingale deflators, this precludes the arbitrage opportunities in the sense of unbounded profit with bounded risk (UPBR) introduced in \cite{KarKar07}. In other words, 
\be\label{noArb}
NUPBR~~ ho\tr lds~~ for~~ every ~~\e\in(-\e_0, \e_0). 
\ee

For the BSDE characterizations, as in \cite{Hao17}, it is important to restrict the admissible consumptions to the ones that are also integrable in a sense made precise below. Thus, one can define
\be\label{defAa}{\widetilde\cA}(x,\e) : = \left\{c\in\cA(x,\e) :~\bE\[\int_0^T 
 c_s^{1 - \frac 1\psi}ds\]<\infty \right\}.
\ee
Formally, in \cite{Hao17}, also $\bE\[c^{1-\gamma}_T\]<\infty$ is imposed. However, 
for every constant $\delta>0$, a consumption plan satisfying $c_T\geq \delta$ satisfies $\bE\[c^{1-\gamma}_T\]< \infty$. 
In particular, the plans such that $\bE\[c^{1-\gamma}_T\]= \infty$  correspond to small values of $c_T$, and thus are suboptimal. By setting the associate $U^c\equiv-\infty$ for every $c$ such that $\bE\[c^{1-\gamma}_T\]= \infty$, one can rule them out. If all consumption plans allow for $\bE\[c^{1-\gamma}_T\]= \infty$, then intuitively, the problem is degenerate. This, however, does not happen if the interest rate $r^0\geq 0$, in which case  constant-valued consumptions are admissible and integrable in the sense above. 


Having ruled out the possibility of $\bE\[c^{1-\gamma}_T\]= \infty$ for all consumption plans, as in the paragraph above, one can provide a sufficient condition for  $\bE\[\int_0^T 
c_s^{1 - \frac 1\psi}ds\]<\infty$ to hold for every $c\in\cA(x,\e)$. It is related to a characterization via the {reverse H\"older inequality in the spirit of  \cite[Proposition 4.5]{MarcelOp} and \cite{Kazamaki}}. 

{The following lemma provides a sufficient condition for $\cA(x,\e) = {\widetilde\cA}(x,\e)$.}
\begin{Lemma}\label{lemConvA}
{Let $\e\in\(-\e_0, \e_0\)$ be fixed and suppose that there exists $D\in\cD(1, \e)$, such that 
\be\label{RevHold}
\bE\[ \int_0^T D^{1 - \psi}_s ds \]<\infty.
\ee
Then, we have  
\be\label{412}
\cA(x,\e) = {\widetilde\cA}(x,\e),\quad x>0.
\ee}
\end{Lemma}
%
\begin{proof}
{Let us fix $\e\in(-\e_0, \e_0)$. 
Then, for every $D\in\cD(1,\e)$, along the lines of  \cite[Proposition 4.2]{Mostovyi2015}, one can show that  
\be\label{414}
\bE\[\int_0^T D_s c_s ds\]\leq 1,\quad for~every~c\in\cA(1,\e).
\ee
Next, let us consider $D\in\cD(1,\e)$, satisfying \eqref{RevHold}. 
Then, for an arbitrary $c\in\cA(1,\e)$, using H\"older's inequality, we get
\be\nn\bs
\bE\[ \int_0^T c^{1-\tfrac 1\psi}_sds\] &= \bE\[ \int_0^T c^{1-\tfrac 1\psi}_sD^{1-\tfrac 1\psi}_sD^{\tfrac 1\psi -1}_sds\]  \\
&\leq C\bE\[ \int_0^T D_sc_sds\]^{{1 - \tfrac 1\psi}}\bE\[\int_0^T D^{1 - \psi}_sds\]^{\frac 1\psi}<\infty,
\end{split}\ee 
for some constant $C\in(0,\infty)$, 
where the last inequality follows from \eqref{RevHold} and \eqref{414}. Therefore, $\bE\[\int_0^T e^{-\delta s } c_s^{1 - \frac 1\psi}ds\]<\infty$, and we conclude that $c\in\widetilde\cA(1,\e)$. 
 }
%

\end{proof}

As pointed out in \cite[Remark 2.2]{MX}, instead of verifying the integrability conditions in \eqref{defAa}, it is enough to check for the optimal consumption stream, $c^*$, that the associated $U^{c^*}$ exists and is of class (D).  A similar argument can be provided for the dual problem below.
 
 With the integrability conditions in \eqref{defAa}, one can restate \eqref{primalProblemOld} as 
 \be\label{primalProblem}
 u(x,\e) = \sup\limits_{c\in{\widetilde\cA}(x,\e)}U^c_0,\quad (x,\e)\in(0,\infty)\times \(-\e_0,\e_0\).
 \ee
We call $u$ - the value function and $U^{\widehat c(x,\e)}$ - the value process if $\widehat c(x,\e)$ is an optimizer in \eqref{primalProblem} for a given pair $(x,\e)\in(0,\infty)\times \(-\e_0,\e_0\)$, provided that such an  optimizer exists. 
Next, following \cite{MX}, let us define
\be\label{defg}
g(d, v) := \delta^\psi \frac{d^{1-\psi}}{\psi - 1}\((1-\gamma)v\)^{1 - \frac{\gamma\psi}{\theta}}-\delta \theta v,\quad d>0,\quad (1-\gamma) v>0,
\ee
and a function $V_T$, the convex conjugate of $U_T$, which is given by 
\be\label{defVT}
V_T(d):= \frac \gamma{1 - \gamma}d^{\frac{\gamma - 1}{\gamma}},\quad d>0.
\ee
Next, for a given pair $(y,\e)\in(0,\infty)\times \(-\e_0,\e_0\)$ and $D\in\cD(y, \e)$, one defines the Epstein-Zin {\it stochastic differential dual} for $D$ to be a process $V^{D}$ satisfying the BSDE
\be\label{defVD}
V_t^{D}  = \bE_t\[V_T(D_T) + \int_t^T g\(D_s, \frac 1\gamma V^{D}_s\)ds \],\quad t\in[0,T].
\ee
Sufficient conditions for the existence of $V^D$ are presented in \cite[Proposition 2.5]{MX}.
We state the family of the dual minimization problems as
\be\label{dualProblemOld}
\inf\limits_{D\in\cD(y, \e)}V_0^D,\quad (y,\e)\in(0,\infty)\times \(-\e_0,\e_0\).
\ee
Similarly to \eqref{primalProblem},
to ensure that for a given state price density $D$, $V^D$ is well-defined, one needs some integrability conditions, and following \cite[Proposition 2.5]{MX}, one can set
\be\label{defDa}
{\widetilde\cD}(y,\e) := \left\{D\in\cD(y,\e):~\bE\[\int_0^T D_s^{1-\psi}ds\]\right\}<\infty.
\ee
Technically in \cite[Proposition 2.5]{MX}, it is also required that $\bE\[D_T^{\frac{\gamma -1}\gamma} \]<\infty$, which however holds in our settings  for every state price density, by an application of Holder's inequality, as ${\frac{\gamma -1}\gamma}\in(0,1)$, and since $D$ (under nonnegative interests rates) is a supermartingale. 
This allows us to restate \eqref{dualProblemOld} as
\be\label{dualProblem}
v(y,\e):=\inf\limits_{D\in{\widetilde\cD}(y, \e)}V_0^D,\quad (y,\e)\in(0,\infty)\times \(-\e_0,\e_0\).
\ee
We conclude this section by pointing out that, by \eqref{RevHold}, if ${\widetilde\cD}(1,\e)\neq \emptyset$, then $\cA(1,\e) = {\widetilde\cA}(1,\e)$, and thus, the convexity of ${\widetilde\cA}(1,\e)$ holds. {Thus, for every $\e\in(-\e_0, \e_0)$, the non-emptiness of the dual feasible set implies the convexity of the primal domain. }
\section{Main results}\label{secMainResults}
\subsection{Model assumptions}{
We will need two assumptions. 
To ensure that the dual problem \eqref{dualProblem} is non-degenerate in a neighborhood of $\e = 0$, we impose the following assumption. 
\begin{Assumption}\label{asPert}
For every $\e\in(-\e_0, \e_0)$, $\widetilde \cD(1, \e)\neq \emptyset$. 
\end{Assumption}
}

The second assumption allows for the additional structure for the base model corresponding to $\e=0$.
\begin{Assumption}\label{asBaseModel}
Let $x>0$ be fixed and suppose that, for $\e=0$, a conjugacy relation in the following sense holds
:
\be\label{conjugacy}
u(x,0) = \min\limits_{\widetilde y> 0}\( v(\widetilde y, 0) + x\widetilde y\) = v( y, 0) + x y,
\ee
for some $y>0$. Further, assume that, for $\e=0$, there exist optimizers $\widehat c(x, 0)$ to \eqref{primalProblem} and $\widehat D(y,0)$ to \eqref{dualProblem}, such that $U^{\widehat c(x, 0)}$ and $V^{\widehat D(y,0)}$ are of class D.
\end{Assumption}
{\subsection*{Sufficient conditions for Assumption \ref{asBaseModel}}
Sufficient conditions for \eqref{conjugacy} are contained in \cite[Section 3]{MX}. Explicit solutions are contained \cite[Theorem 2.14]{Hao17},  
 see also \cite{Kraft17}, where optimal strategies are obtained in Markovian settings. To be more precise, \cite{Kraft17} contains the explicit solution for the primal problem \eqref{primalProblem}, and the optimal state price density could be identified via the utility gradient approach, following e.g., \cite{DuffieSkiadas94}. 
\begin{Proposition}\cite[Theorem 3.6]{MX}
Suppose that $\gamma\psi\geq 1$, $\psi>1$, or $\gamma\psi\leq 1$, $\psi<1$ and the processes $r^0$, $(\mu^0)^\top 
\(\(\sigma^0\)^{\top}\)^{-1}\(\sigma^0\)^{-1}
\mu^0$ are bounded. Then \eqref{conjugacy} holds.
\end{Proposition}
For models with unbounded market price of risk, we refer to \cite[Section 3.4]{MX} for the exact conditions that guarantee Assumption \ref{asBaseModel}. In a Markov setting, we refer to \cite[Theorem 5.1]{Kraft17}, where, in the one-dimensional stock prices process and a factor model for the dynamics of both riskless and risky assets, boundedness of $\mu^0$ and $r^0$ as well as boundedness away from $0$ and $\infty$ of $\sigma^0$, guarantee that  \cite[Theorem 3.6]{MX} applies.}  

To analyze the behavior of the primal and dual problems under perturbations, we introduce a family of     $\bR^{n}$-dimensional processes $\lambda^\e$, defined by  
\be\label{deflambda}
\lambda^\e_t :=
  \((\sigma^0_t)^\top\)^{-1}\(\(\sigma^\e_t\)^{-1} \mu^\e_t - \(\sigma^0_t\)^{-1} \mu^0_t\),
\quad t\in[0,T],\quad \e\in(-\e_0,\e_0).
\ee
We also set $R:=(R^1,\dots,R^n)$, where 
\be\label{defRe}
\quad dR^{i}_t = \mu^{0,i} dt+ \sum\limits_{j = 1}^n\sigma_t^{0, i, j}dW^{\rho, j}_t,\quad R^{ i}_0 = 0, \quad i\in\{1, \dots, n\},
\ee
Along the lines of \cite{MostovyiNumeraire}, let us introduce the family of processes $N^\e$, $\e\in(-\e_0,\e_0)$, given via
\be\label{defNe}
 d{N^\e}_t = N^\e_t\((r^0_t - r^\e_t) dt - \lambda^\e_t {dR_t}\), \quad t\in[0,T], \quad N^\e_0 = 1,\quad \e\in(-\e_0,\e_0).
 \ee
 
 We recall that $\kappa$ is given by  $\kappa_t = t + 1_{\{T\}}(t)$, $t\in[0,T]$.  Let  $\bL^0(d\kappa\times\bP)$ 
 be  the linear space of (equivalence classes of) real-valued
measurable processes on the stochastic basis $(\Omega, \cF, (\cF_t)_{t\in[0,T]}, \bP)$ equipped with the topology of convergence in measure $(d\kappa\times\bP)$. 
%

\subsection{Stability theorems}
The first theorem establishes convergence of the value functions. 
\begin{Theorem}\label{mainThm1}
Let $x>0$ be fixed, {$\gamma> 1$ and $\psi>1$} in \eqref{deff}. Let us further suppose that  {Assumptions \ref{asPert} and \ref{asBaseModel}} hold and 
for every $\e\in(-\e_0, \e_0)$, $\sigma^\e$ is invertible, {$\lambda^\e$ appearing in \eqref{deflambda}  is $R$-integrable} 
 and 
$\lim\limits_{\e\to 0} N^\e = 1,~ {in ~measure~(d\kappa\times \bP)}.$

Then for every $\e\in\(-\e_0,\e_0\)$, we have 
\begin{enumerate}[(i)]\item
 the value functions are finite-valued, that is 
\be\label{finValueuv}
u(z,\e)\in\bR\quad and\quad v(z, \e)\in\bR,\quad (z,\e)\in(0,\infty) \times\(-\e_0,\e_0\);
\ee
\item 
the value functions converge 
\be\label{convergenceu}\bs
\lim\limits_{(x',\e)\to (x,0)}u(x',\e) = u(x,0),\quad x >0,
\end{split}
\ee
\be\label{convergencev}\bs
\lim\limits_{(y',\e)\to (y,0)}v(y',\e) = v(y,0),\quad y >0;
\end{split}
\ee
\item for every $(x,\e)\in(0,\infty)\times(-\e_0,\e_0)$,
 there exists a unique optimizer to \eqref{primalProblem}.
\end{enumerate}
\item  
\end{Theorem}
\begin{Remark}\label{remFinValue}
For the problem in \eqref{primalProblem},    a condition of the finiteness of the value functions condition is typically imposed. In the present settings, as we deal with non-positive value functions $u(x, \e)$ finiteness from above (by zero) holds. For the finiteness from below,
\be\label{finValue}
 u(x, \e)>-\infty,\quad (x,\e)\in(0,\infty)\times(-\e_0,\e_0),
\ee
{we remark that this also holds as $r^0\geq 0$, and thus $c\equiv \frac x{T+1}$ is an admissible consumption for the initial wealth $x$, for which one can use comparison results for BSDEs to show that the value function is finite-valued. Similar arguments can be employed to show the finiteness of $v(y,\e)$, as it is also bounded by zero from above, and by $\(u(1, \e) - y\)$ from below. }
\end{Remark}


 The next theorem addresses the convergence of the optimizers. The assumptions of Theorem \ref{mainThm1}, ensure that for every $(x,\e)\in(0,\infty)\times (-\e_0, \e_0)$, there exists a unique $\widehat c(x,\e)$, such that $u(x,\e) = U_0^{\widehat c(x,\tr \e)}$, and that $u(x,\e)$ is finite-valued for every such $(x,\e)$. To prove convergence of the optimizers, we need to ensure finiteness for the value processes in the sense below.

\begin{Theorem}\label{mainThm2}
Let $x>0$ be fixed. Let the conditions of Theorem \ref{mainThm1}  
 hold and suppose that there exists $\e'>0$, such that 
\be\nn\label{asFin}\esssup\limits_{(x,\e)\in B_{\e'}(0,0)}\widehat c(x,\e) \in\bL^0(d\kappa\times\bP),~~
\essinf\limits_{(x,\e)\in B_{\e'}(0,0)}U^{\widehat c(x,\e)} \in\bL^0(d\kappa\times\bP),
\ee
where $B_{\e'}(0,0)$ is a Euclidean ball of radius $\e'$ in $\bR^2$. 
 
We then have that 
\be\label{convOpt}\bs
\lim\limits_{(x',\e)\to (x,0)} \widehat c(x',\e)= \widehat c(x,0),
\end{split}
\ee
where the convergence is in measure $(d\kappa\times\bP)$. 
\end{Theorem}

{
Let us recall that under the conditions of Theorem \ref{mainThm1}, the existence and uniqueness of the optimizer to \eqref{primalProblem} follows from Theorem \ref{mainThm1}, item $(iii)$. Let us also recall that, for a given nonnegative consumption stream $c$,  $U^c$ was defined in \eqref{defUEZ}.
The following theorem establishes the convergence of the indirect utility processes. 
}

{
 \begin{Theorem}\label{thmConvU}
 Let $x>0$ be fixed. Then, under the conditions of Theorem \ref{mainThm2}, we have 
$$
\lim\limits_{(x',\e) \to (x,0)}U^{\widehat c(x', \e)} =  U^{\widehat c(x,0)},\quad ucp.$$
 \end{Theorem}
 }
 {
 Next, for a fixed $x>0$ and $\e = 0$, let $y>0$ be as in Assumption \ref{asBaseModel} and suppose that the dual minimizer has the form    
 \be\label{defD*}
 \widehat D_t(y, 0) = 
C 
 \exp\left( \int_0^t \partial_uf\left( \widehat c_t(x,0), U^{\widehat c(x,0)}_t\right)ds\right) \partial_cf\left(\widehat c_t(x,0), U^{\widehat c(x,0)}_t\right),\quad t\in[0,T],
 \ee
 for some constant $C>0$ and 
\be\label{mart}
 X^{\widehat c(x,0)}\widehat D(y,0) + \int_0^\cdot \widehat D_s(y,0) \widehat c_s(x,0) ds\quad {\rm is ~a~ martingale,}
\ee
where $X^{\widehat c(x,0)}$ is the wealth process starting from $x$ financing $\widehat c(x,0)$ (given by \eqref{Xe} at $\e=0$).
We note that sufficient conditions for \eqref{defD*} and \eqref{mart} are similar to the ones for Assumption \ref{asBaseModel} to hold; see the discussion after Assumption \ref{asBaseModel}. In particular, both \eqref{defD*} and \eqref{mart}  hold if the market price of risk $(\mu^0)^\top 
\(\(\sigma^0\)^{\top}\)^{-1}\(\sigma^0\)^{-1}
\mu^0$ process is bounded as well as $\gamma>1$ and $\psi>1$. Then, the conditions of \cite[Theorem 3.6, p. 1002]{MX}, apply and \cite[Corollary 3.7, p. 1002]{MX} implies \eqref{defD*} and \eqref{mart}. Furthermore, representation \eqref{defD*} for the optimal state-price density goes back to \cite{DuffieSkiadas94} and is known as the utility gradient approach.  
\begin{Theorem}\label{lemWealthPr}
Let $x>0$ be fixed. 
Let the assumptions of Theorem \ref{mainThm2}, \eqref{defD*}, and \eqref{mart} hold. Then, $\lim\limits_{(x',\e)\to(x,0)} X^{\widehat c(x',\e)} = X^{\widehat c(x,0)} $, in measure $(d\kappa \times \bP)$, where $X^{\widehat c(x',\e)}$ is {any} wealth process financing $ \widehat c(x',\e)$ starting from the initial capital $x'$ in the market where the traded assets are given by \eqref{Se}. Furthermore, if $\lim\limits_{\e\to 0}N^\e= 1$, $ucp$, then $\lim\limits_{(x',\e)\to (x,0)}X^{\widehat c(x',\e)} = X^{\widehat c(x,0)} $,  ucp. 
\end{Theorem}
}
\section{On the integrability condition on perturbations}\label{secInteg}
{Let us revisit Assumption \ref{asPert}. 
For a fixed $\e\in(-\e_0, \e_0)$, in order for $\widetilde \cD(1, \e)\neq \emptyset$, where $\widetilde \cD(1, \e)$ are defined in \eqref{defDa}, there must exist a supermartingale state price density $D^\e\in\cD(1, \e)$, such that 
\be\label{441}
\bE\[\int_0^T (D^\e_s)^{1-\psi}ds\]<\infty.
\ee 
The natural candidate for \eqref{441} to hold is to check whether the minimal state price density given by \eqref{eqMinStatePriceDensity} satisfies the integrability condition \eqref{441}. }
{
Another sufficient condition for Assumption \ref{asPert} to hold is given by 
\be\nn
\bE\[\int_0^T \(\widehat D_s(y, 0)N^\e_s\)^{1- \psi}ds\]<\infty,\quad \e\in(-\e_0,\e_0),
\ee
where $\widehat D_s(y, 0)$ is the dual minimizer at $(y,0)$, which exists by Assumption \ref{asBaseModel} and $N^\e$ are given by \eqref{defNe}.}
Condition 
\eqref{441} is the only integrability condition needed on the perturbations to ensure that the dual domain incorporating the additional integrability for perturbed models is non-empty, that is: ${\widetilde\cD}(1,\e)\neq\emptyset$, for $\e\neq 0.$
Perhaps the most surprising feature in our analysis (at least for the authors) was  that other than \eqref{441}
, no further integrability needs to be imposed. This can be explained as follows, where the key is in the utility maximization considerations. It is well-known that for the expected utility maximization from terminal wealth, the key role is played by the finiteness of the value functions, see \cite{KS}, where the finiteness of the primal value function (from above) is assumed, and \cite{KS2003}, where the finiteness of the dual value function (from above) is required. To be more precise, both conditions require the value functions to be less than $\infty$ (in \cite[Theorem 2.2]{KS}, under the asymptotic elasticity). 
In \cite{Mostovyi2015}, the finiteness of both primal and dual value functions (from below and above) is introduced and proven to be necessary and sufficient for the standard assertions of the utility maximization theory in the case of additive and stochastic utility. 

In the present setting, in view of the choice of the parameters $\gamma>1$ and $\psi>1$, we obtain that the associated value function is negative-valued. This follows from the analysis of the associated BSDEs as in {Lemma \ref{lemFinValue}}. In particular, the base model exhibits  a finiteness conditions for both the primal and dual value functions. For the perturbed models, as $\gamma$ and $\psi$ do not change, we still obtain that the primal value function is negative-valued, and the dual one too. Here, the primal gives a lower bound for the dual via the conjugacy relations. Thus the blow-up to $\infty$ is not possible under perturbations of the models. In turn, the blow-up to $-\infty$ is also not possible, as \eqref{primalProblem} is a {\it maximization} problem, and thus finiteness of the base model guarantees that we do not have a blow-up as long as the processes $N^\e$'s  appearing in \eqref{defNe} are well-defined, and without any further integrability conditions needed on this family. This situation can be compared to the counterexample in \cite{MostovyiNumeraire}, where blow-up does happen for a particular form of perturbations,  as the utility function there can take positive values.

The connection between the last two paragraphs can be further illustrated by the case  of $\gamma = \frac 1 \psi$ (going outside the scope of the analysis in this paper). 
Then the problem \eqref{primalProblem} reduces to the one with an {\it additive utility}, given by  
$$U^c_0 = \bE\[ \int_0^T \delta e^{-\delta s}\frac {c_s^{1-\gamma}}{1-\gamma} ds + e^{-\delta T}\frac {c_T^{1-\gamma}}{1-\gamma}\].$$
In this case, and with $\gamma>1$, the value function is negative-valued, and so is the dual one, thus precluding the blow-up to $\infty$. The blow-up to $-\infty$ is not possible by the feasibility of the constant-valued consumptions, which also gives a  lower bound for the dual problem.
\section{Proofs}
\label{secProofs}
\subsection{Preliminary Results} We begin with the following structural lemma. 
\begin{Lemma}\label{lemAdm}Let the conditions of Theorem \ref{mainThm1} hold, let $x>0$ be fixed, and $y>0$ be given through \eqref{conjugacy}. Then, for every $\e\in(-\e_0,\e_0)$, we have
 \be\label{eAdm}
 \bs
c^\e: &=
 \widehat c(x,0)\frac 1{N^\e} 
\in\cA(x,\e) ,\quad x>0,\\ 
D^\e:&=\widehat D(y, 0)N^\e\in \cD(y, \e),\quad y>0,
\end{split}
\ee
where $\widehat c(x,0)$ and $\widehat D(y, 0)$ are the optimizers, for $\e=0$, to \eqref{primalProblem} and \eqref{dualProblem}, respectively. 
\end{Lemma}
\begin{proof}
First, we observe that for every $\e\in(-\e_0,\e_0)$, the process $N^\e$ is progressively measurable by \cite[Proposition 1.13, p. 5]{KaratzasShreve1}. 
Now, the assertion of the lemma follows from  It\^o's lemma. 
\end{proof}
Let us introduce some notations used in this section's remaining part.
\begin{itemize}
\item Let $\cS^2$ be the space of one-dimensional continuous adapted processes $(Y_t)_{t\in[0,T]}$ such that $\bE\[ \sup\limits_{t\in[0,T]}|Y_t|^2\]<\infty$.
\item Let $\cS^\infty = \left\{Y\in \cS^2:~||\sup\limits_{t\in[0,T]}|Y_t|||_{\infty}<\infty\right\}$. 
\item Let $\cM^2$ denote the space of predictable multidimensional  processes $(Z_t)_{t\in[0,T]}$, such that $\bE\[\int_0^T |Z_t|^2 dt\]<\infty$. 
\end{itemize}

 With $f$ given in \eqref{deff}, let us consider the BSDE
 \be\label{BSDEU}
 U^c_t= \frac{c_T^{1-\gamma}}{1 - \gamma} + \int_t^T f(c_s, U^c_s)ds - \int_t^T Z_s^c dB_s,\quad t\in[0,T].
 \ee
Next, with the transformation 
\be\label{transfY}
(Y, Z) := e^{-\delta \theta t} (1- \gamma) (U^c, Z^c),\quad t\in[0,T],
\ee
 we obtain a BSDE for $(Y, Z)$ of the form
\be\label{BSDEY}
Y_t = e^{-\delta\theta T}c_T^{1-\gamma} + \int_t^T F(s, c_s, Y_s) ds - \int_t^T Z_s dB_s,\quad t\in[0,T],
\ee
where,  for $\theta<0$, $
F(t,x, y) := \delta \theta e^{-\delta t} x^{1 - \frac 1\psi} y^{1 - \frac 1 \theta}\leq 0$   is monotonically decreasing in $y$. 
\begin{Lemma}\label{lemFinValue}
Under the conditions of Theorem \ref{mainThm1}, for every  $(z,\e)\in(0,\infty)\times(-\e_0, \e_0)$, $u(z, \e)$ and $v(z, \e)$ are finite-valued. 
\end{Lemma}
\begin{proof}
Let us fix $(z,\e)\in(0,\infty)\times(-\e_0, \e_0)$. From Lemma \ref{lemAdm}, we deduce that $\frac zx\widehat c(x,0)\frac 1{N^\e} \in\cA(z, \e)$. Therefore, for a fixed $m\geq 1$, \be\label{431}
c:=\frac{z}{x+\tfrac 1m}{\frac 1m}\vee\(\widehat c(x,0)\frac 1{N^\e}\)\wedge m \in\cA(z, \e).
\ee In particular, we have   
 $$\bE[(c_T)^{1-\gamma}]<\infty.$$ 
 Next, with $F^k(t, c_t, y)  : = \delta \theta e^{-\delta t} c_t^{1- \frac 1 \psi}(|y| \wedge k)^{1- \frac 1 \theta}$ (note that the process $c$ is bounded from above by $m$ here), let us consider 
$$Y^k_t = e^{-\delta\theta T}c_T^{1-\gamma} + \int_t^T F^k(s, c_s, Y^k_s) ds - \int_t^T Z^k_s dB_s,\quad t\in[0,T],\quad k\in\bN.$$
This is a BSDE with a Lipschitz generator and a bounded terminal condition. Therefore,  by 
\cite[Theorem 5.1]{CohenElliot2}, this BSDE admits a unique solution $(Y^k, Z^k)\in\cS^2\times \cM^2$. {Furthermore, as $c$ in \eqref{431} is bounded away from $0$ and $\infty$, we have $\frac 1{\bar C}\leq Y^k_T\leq \bar C$, for some constant $\bar C>0$. As, additionally, $F^k$ is non-positive-valued, using the comparison result for BSDEs  
\cite[Theorem 2.4, p. 517]{Pard}, one can show that $Y^k$ takes values in $[0,\tr {\bar C}]$. Therefore, with $c$ in \eqref{431} and the associated $Y$ given via \eqref{BSDEY}, for every $k\geq \bar C$, $F^k(t, c_t, Y^k) = F(t, c_t, Y^k)$, $t\in[0,T]$, $\Pas$. As a result, $(Y, Z) : = (Y^k, Z^k)$ is a solution to \eqref{BSDEY} (for $c$ given in \eqref{431}). 
}

Changing variables back to $(U^c, Z^c)$, {that is from \eqref{transfY}, and with $$(U^c_t, Z^c_t) := \frac{e^{\delta \theta t} }{1- \gamma}(Y_t, Z_t), \quad t\in[0,T],$$} one can show that this pair satisfies \eqref{BSDEU} and further, {following the proof of \cite[Proposition 2.2]{Hao17}}, that $U^c$ satisfies \eqref{defUEZ}, $U^c$ is non-positive-valued and is bounded away from $-\infty$. As in \eqref{primalProblem}, we take the supremum over all admissible consumptions, $u(z, \e) \geq U^c_0$ (for $c$ as above). Next, {also similarly to the proof of \cite[Proposition 2.2]{Hao17} and relying on the localization technique from \cite{BriandHu}}, one can see that for every admissible consumption, $U^c_0\leq 0$. We conclude that $u(z,\e)\leq 0$ and is finite-valued  for every $\e\in(-\e_0,\e_0)$. 
Now, by \cite[Theorem 2.7]{MX}, $v(z,\e) \geq u(x,\e) -xz$, $(x,z)\in(0,\infty)^2$, and thus $v(z,\e)$ is bounded away from $-\infty$. 
Furthermore, similarly to showing that $u(z,\e)\leq 0$, one can show that  $v(z,\e)\leq 0$. We conclude that $v(z,\e)$ is finite-valued.  
\end{proof}
\begin{Lemma}\label{lemusc}Let $x>0$ be fixed. Then, 
under the conditions of Theorem \ref{mainThm1}, 
we have 
\be\label{bsdeConv}
\liminf\limits_{(x',\e) \to (x,0)} u(x',\e) \geq u(x,0).
\ee
\end{Lemma}
\begin{proof}
 Let us fix $\e'>0$ and let $\widehat c(x,0)$ be the optimizer to \eqref{primalProblem} at $(x,0)$, which belongs to ${\widetilde\cA}(x,0)$, as this results from Assumption \ref{asBaseModel}. Let $(x_k,\e_k)$, $k\in\bN$, be a sequence  which converges to $(x,0)$ and such that 
\be\label{72710}
\lim\limits_{k\to\infty} u(x_k,\e_k) = \liminf\limits_{(x',\e) \to (x,0)} u(x',\e).\ee
For $c = \widehat c(x,0)$, let us 
consider the BSDE \eqref{BSDEY} (which  is related to \eqref{BSDEU} via \eqref{transfY}).  
As, by Assumption \ref{asBaseModel}, $U^c$ is of class D, one can show  (see the discussion in \cite[Remark 2.2]{MX}) that $c\in{\widetilde\cA}(x,0)$. 
Furthermore, as established in the proof of \cite[Proposition 2.2]{Hao17}, {(for $\gamma, \psi >1$)}  
 \eqref{BSDEY} admits a {\it unique} solution $(Y, Z)$, such that $Y$ is continuous, strictly positive, and of class  D,  with $\int_0^T|Z_t|^2 dt <\infty$, $\Pas$. Moreover, $U^c := e^{\delta\theta t}Y_t \frac 1{1-\gamma}$, $t\in[0,T]$, satisfies \eqref{BSDEU} and \eqref{defUEZ}. 
Next, using the approximation procedure as in step 2 of the proof of \cite[Proposition 2.2]{Hao17}, 
 one can show that there exists $n_0\in\bN$, such that \be\label{7232}
|Y^n_0 - Y_0|<\frac {\e'}3, \quad n\geq n_0,
\ee
where $Y^n$ solves
$$Y^n_t = (e^{-\delta \theta T} c_T^{1-\gamma} )\wedge n + \int_t^T F(s, c_s, Y^n_s) ds - \int_t^T Z^n_s dB_s,\quad  t\in[0,T].$$

The latter BSDE admits a solution $(Y^n, Z^n)\in\cS^\infty\times \cM^2$, where using the comparison argument, one can show that $0\leq Y^n \leq n$, and $Y^n = \downarrow \lim\limits_{m\to\infty} Y^{n, m}$, $n\in\bN$, where $Y^{n,m}$ solves
\be\label{7235}
Y^{n,m}_t = (e^{-\delta \theta T} c_T^{1-\gamma} )\wedge n + \int_t^T F^m(s, c_s, Y^{n, m}_s)ds  - \int_t^T Z^{n,m}_s dB_s,\quad  t\in[0,T],
\ee with $F^m(t, c_t, y)  : = \delta \theta e^{-\delta t} (c_t^{1- \frac 1 \psi}\wedge m)(|y| \wedge m)^{1- \frac 1 \theta}.$
Likewise, one can show that  $0\leq Y^{n, m}\leq n$. Therefore, for $m\geq n$, we obtain 
\be\label{7231}\bs
F^m(t, c_t, Y^{n, m}_t)   &= \delta \theta e^{-\delta t} (c_t^{1- \frac 1 \psi}\wedge m)(Y^{n, m}_t \wedge m)^{1- \frac 1 \theta} \\
&= \delta \theta e^{-\delta t} (c_t^{1- \frac 1 \psi}\wedge m)(Y^{n, m}_t)^{1- \frac 1 \theta}.\end{split}
\ee
For every $n\in\bN$, one can show that  $\lim\limits_{m\to\infty} \sup\limits_{t\in[0,T]}|Y_t^n - Y_t^{n, m}| = 0$ in probability $\bP$, and thus, we deduce that there exists $m'(n)\in\bN$, such that 
\be\label{7233}|Y^n_0 - Y^{n, m}_0|<\frac {\e'}{3},\quad n\in\bN, \quad m\geq m'(n).\ee

 It follows from Lemma \ref{lemAdm} that the process $c^{k}= \widehat c(x,0)\frac 1{N^{\e_k}}\in\cA(x, \e_k)$. Next, 
for every $M_1>0, M_2>0$ the process $\bar c^k$ defined as 
$\bar c^k_t := \frac{x_k}{x + \frac 1{M_1}}\frac 1{M_1}\vee c^{k}_t\wedge M_2$,
$t\in[0,T]$, satisfies 
$\bar c^k \in{\widetilde\cA}(x_k, \e_k).$
Now let us consider the sequence of BSDEs 
\be\label{7236}
\bar Y^k_t =  e^{-\delta \theta T} (\bar c^k_T)^{1-\gamma}   +\int_t^T F(s, \bar c^k_s, \bar Y^k_s)ds - \int_t^T \bar Z^k_s dB_s,\quad t\in[0,T],\quad k\in\bN.
\ee
\cite[Theorem 5.1]{CohenElliot2} ensures that there exists a unique solution to BSDE \eqref{7236}, $\(\bar Y^k, \bar Z^k \)\in \cS^2\times \cM^2$. Further, by replacing $F$ with $F^k$ as in the previous step, and using the {comparison for BSDEs ({with Lipschitz generator}) results, see, e.g., \cite[Theorem 2.4]{Pard}}, we deduce that the first component of the solution is in $\cS^\infty$.

Let us consider \eqref{7235} and \eqref{7236}. These are BSDEs with bounded terminal conditions and Lipschitz generators. Therefore, the stability of BSDEs, as in   \cite[Theorem 19.1.6, p. 472]{ElliotCohen}, implies  that, for some $n$ satisfying \eqref{7232} and for $m = n(m)$ satisfying \eqref{7233}, one can first pick $M_1$ and $M_2$ and then  $k_0$, such that \be\label{7238}
|\bar Y^k_0 - Y^{n, m}_0|<\frac {\e'}3, \quad k\geq k_0.
\ee
Comparing \eqref{7232}, \eqref{7233}, and \eqref{7238}, we deduce that 
$$|\bar Y^k_0 - Y_0|<\e', \quad k \geq k_0.$$
Therefore, as $\bar c^k \in{\widetilde\cA}(x_k, \e_k)$, via \eqref{transfY}, we obtain 
$$\liminf\limits_{k\to\infty}u(x_k,\e_k) \geq \liminf\limits_{k\to\infty}\frac{\bar Y^k_0}{1 - \gamma} \geq \frac{ Y_0}{1 - \gamma} - \frac{\e'}{|1 -\gamma|} = u(x,0) - \frac{\e'}{|1 -\gamma|}.$$
Consequently, as $\e'$ is arbitrary,  via \eqref{72710}, we deduce that \eqref{bsdeConv} holds. 


\end{proof}
The next lemma establishes a result similar to Lemma \ref{lemusc} for the dual value function. The proof is similar to the proof of Lemma \ref{lemusc}, {so we only outline the main steps.} 
\begin{Lemma}\label{lemvsc}
Under the conditions of Theorem \ref{mainThm1}, 
we have 
\be\label{bsdeConvDual}
\limsup\limits_{(y',\e) \to (y, 0)} v(y',\e) \leq v(y,0).
\ee
\end{Lemma} 
\begin{proof}{
Let us consider a sequence $(y_k, \e_k)$, $k\in\bN$, convergent to $(y, 0)$ and such that 
$$\lim\limits_{k\to\infty}v(y_k, \e_k) = \limsup\limits_{(y', \e)\to (y, 0)}v(y, 0).$$
By Assumption \ref{asPert}, for every $k\in\bN$, there  $\widetilde D^k\in\widetilde \cD(1, \e_k)$, that is, $\widetilde D^k$, such that  
\be\label{443}
\bE\[\int_0^T (\widetilde D^k_s)^{1-\psi}ds\]<\infty,\quad k\in\bN.
\ee
Let us set 
\be\nn
D^k = (1 - (-1 \vee\e_k \wedge1))\frac {y_k}y \widehat D(y, 0)N^{\e_k} + (-1 \vee\e_k \wedge1)y_k\widetilde D^k, \quad k\in\bN.
\ee
Then, as $N^{\e_k} \to 1$, in measure $(d\kappa\times\bP)$, by the assumption of Theorem \ref{mainThm1}, we deduce that $D^k\to \widehat D(y, 0)$, in measure $(d\kappa\times\bP)$. Moreover, it follows from \eqref{443}, and since $1 - \psi<0$, that $D^k\in \widetilde \cD(y_k, \e_k)$, $k\in\bN$. Next, applying the approximation procedure entirely similarly to Lemma \ref{lemusc}, we obtain the assertion of the lemma.}
\end{proof}

We now show the concavity of $U^c_0$ in $c$ in the following sense, which is closely related to the notion of strong concavity.
\begin{Lemma}\label{lemConv}
Let us suppose that $c'$ and $c''$ are in $\bigcup\limits_{(x,\e)\in(0,\infty)\times\(-\e_0, \e_0\)}{\widetilde\cA}(x, \e)$  and are such that  
\be\label{6171}
(d\kappa\times\bP)\[ |c' - c''|\geq \delta, ~c'+c''\leq \frac 1\delta\]\geq \delta, \quad for~some~\delta>0.
\ee Further, let us suppose that for a given constant  $\lambda \in(0,1)$, we have
 $$c:= \lambda c' + (1-\lambda)c''\in \bigcup\limits_{(x,\e)\in(0,\infty)\times\(-\e_0, \e_0\)}{\widetilde\cA}(x, \e).$$ 

Then, there exists a constant $\bar\eta  >0$, such that  
\be\label{6172}\lambda U^{c'}_0 + (1-\lambda) U^{c''}_0 +\bar\eta\leq U^{c}_0.\ee   
\end{Lemma}
\begin{proof}
Let us show that  
\be\label{6181}\lambda Y^{c'}_0 + (1-\lambda) Y^{c''}_0 -\eta_0\leq Y^{c}_0,
\ee
where $Y$'s satisfy \eqref{BSDEY} with respective $c$'s and $\eta_0$ is some positive constant. 
As the generator of $Y$ is not {\it jointly} concave in $(c, Y, Z)$, with $p:= 1- \frac 1 \psi$, 
 following \cite{Hao17}, one can set  
$$\(\bY, \bZ\): = \frac {1}{p}\(Y^{\frac 1\theta},\frac 1\theta Y^{\frac 1\theta - 1} Z\).$$
Then $\bY$ satisfies 
\be\label{eqnbY}
\bY_t = e^{-\delta T}\frac {c^p_T}{p}+ \int_t^T \(\delta e^{-\delta s}\frac {c_s^p}{p} + \frac 12 (\theta - 1)\frac {\bZ^2_s}{\bY_s}\)ds - \int_t^T\bZ_s dB_s,\quad t\in[0,T],
\ee
where the generator is {\it jointly} concave in $(c, \bY, \bZ)$ when $\theta<1$.

For $\bY'$ and $ \bY''$, let $\d \bY := \lambda \bY' + (1-\lambda) \bY''$, $\d c: = \lambda c' + (1-\lambda) c''$, and $\d \bZ: = \lambda \bZ' + (1-\lambda) \bZ''$. We observe that 
$$\Delta Y := (p\d\bY)^\theta \quad and \quad \Delta Z:= (1-\gamma)\(p \d\bY\)^{\theta - 1}\d\bZ,$$
satisfy
\be\nn\bs
\Delta Y_t = (p\d\bY_T)^\theta &+ \int_t^T\(\delta \theta e^{-\delta s}(\Delta c_s)^{p} - (1-\gamma)A_t\)\Delta Y_s^{1-\frac 1\theta} ds\\
& -\int_t^T\d Z_sdB_s,\quad t\in[0,T],
\end{split}\ee
where 
\be\label{7281}\bs
A_t = &\frac{\delta e^{-\delta t}}{p} \((\d c_t)^p -\lambda (c'_t)^p - (1- \lambda)(c''_t)^p\) 
\\
&+ \tfrac 12 (\theta - 1)\(\frac {\d\bZ^2_t}{\d\bY_t}  - \lambda \frac {{\bZ'}^2_t}{\bY'_t} - (1 - \lambda) \frac {{\bZ''}^2_t}{\bY''_t} \)\geq 0, 
\end{split}
\ee
as both terms on the right-hand side are nonnegative by the joint concavity of the generator to \eqref{eqnbY}. 
Additionally, the function $x\to x^p$, $x>0$ is strictly concave. Therefore, on $\{c'_t + c''_t\leq \frac 1\e ,|c'_t - c''_t|\geq\e\}$, we have 
 $\((\d c_t)^p -\lambda (c'_t)^p - (1- \lambda)(c''_t)^p\)\geq \delta_1$, for some constant $\delta_1>0$, 
which depends only on $\e$ and $\lambda$ in the statement of the lemma (also on $p=1- \frac 1 \psi$, but $\psi$ is fixed throughout the paper). 
 Similarly, 
we  
obtain  
\be\nn
p e^{\delta T}\D \bY_T \leq (\D c_T)^p - \delta_11_{\{c'_T + c''_T \leq \frac 1\e , |c'_T - c''_T|\geq\e\}}.
\ee
Therefore,  for some constant $\delta_2>0$, we have 
\be\label{7282}
\D Y_T
\geq e^{-\delta\theta T} (\D c_T)^{1-\gamma} + \delta_21_{\{c'_T + c''_T \leq \frac 1\e ,|c'_T - c''_T|\geq\e\}}.
\ee
  It follows from 
\eqref{7281} and \eqref{7282},   that $(\D Y, \D Z)$ is a supersolution to 
$$ Y^{\D c}_t = e^{-\delta\theta T}(\D c_T)^{1-\gamma} + \int_t^TF(s, \D c_s,Y^{\D c}_s) ds -\int_t^T  Z^{\D c}_sdB_s,\quad t\in[0,T].$$
Setting $\xi_t:=- (1-\gamma)A_t\Delta Y_t^{1-\frac 1\theta}$, $t\in[0,T)$, $\xi_T := \D Y_T - Y^{\D c}_T
$, we observe that 
 $\xi_t\geq 0$, $t\in[0,T]$,  and, moreover, for some constant $\widetilde \delta_1>0$, we have
\be\label{8141}\bs
\xi_t&\geq \widetilde \delta_1\Delta Y_t^{1-\frac 1\theta}1_{\{c'_t + c''_t \leq \frac 1\e ,|c'_t - c''_t|\geq\e\}}, \quad t\in[0,T), \\
 \xi_T &\geq\delta_21_{\{c'_T + c''_T \leq \frac 1\e ,|c'_T - c''_T|\geq\e\}}.
\end{split}
\ee
We stress here that $\widetilde \delta_1$ and $\delta_2$ depend only on $\lambda$ and $\e$ in the statement of the lemma.

Further, let us define $\eta_t := \D Y_t - Y^{\D c}_t$, and $\zeta_t: = \D Z_t - Z^{\D c}_t$, $t\in[0,T]$, we deduce that 
\be\label{7283}\bs
\eta_t = \xi_T + \int_t^T\left\{\(\delta \theta e^{-\delta s}(\Delta c_s)^{p} - (1-\gamma)A_t\)\Delta Y_s^{1-\frac 1\theta} -F(s, \D c_s,Y^{\D c}_s) \right\}ds \\
- \int_t^T \zeta_s dB_s.
\end{split}
\ee
Let us rewrite the latter generator as 
\be\nn\bs
&\(\delta \theta e^{-\delta s}(\Delta c_s)^{p} - (1-\gamma)A_t\)\Delta Y_s^{1-\frac 1\theta} -F(s, \D c_s,Y^{\D c}_s) \\
=&- (1-\gamma)A_t\Delta Y_s^{1-\frac 1\theta} + F(s, \D c_s,\D Y_s)  -F(s, \D c_s,Y^{\D c}_s).
\end{split}
\ee
Setting $\alpha_t := \frac {F(t, \D c_s,\D Y_t)  -F(t, \D c_t,Y^{\D c}_t)}{\eta_s}1_{\{\eta_t\neq 0\}}$,  $t\in[0,T]$, we can rewrite \eqref{7283} as 
$$
\eta_t = \xi_T + \int_t^T (\alpha_s \eta_s + \xi_s)ds - \int_t^T \zeta_s dB_s,
$$
With $\Gamma_t := \exp\( \int_0^t \alpha_sds\)$, we get 
\be\nn
\eta_t = \frac 1 {\Gamma_t}\bE_t\[ \Gamma_T\xi_T + \int_t^T \xi_s \Gamma_s ds\].
\ee
In particular, at $t= 0$, we get 
\be\label{7285}
\eta_0 =\bE\[ \Gamma_T\xi_T + \int_0^T \xi_s \Gamma_s ds\] =\bE\[  (\Gamma \xi)\cdot \kappa_T\].
\ee
As both  $\D Y$ and $Y^{\D c}$ are finite-valued, $\Gamma >0$. Next, as 
from \eqref{8141} and \eqref{7285}, we deduce the strict positivity of $\bar U_0$ (by the strict comparison 
  and  \eqref{6171}) that  
  \be\label{7286}
\D Y_0 - Y^{\D c}_0 = \eta_0>0 .
\ee

 Moreover, as in \cite[equation (A.6), p. 247]{Hao17}, we get 
 \be\label{7287}
 \D Y_t\leq \lambda Y^{c'}_t + (1-\lambda)Y^{c''}_t,\quad t\in[0,T],\quad \Pas.
 \ee
 Combining \eqref{7286} and \eqref{7287}, we deduce that 
 $$\lambda Y^{c'}_0 + (1-\lambda)Y^{c''}_0 \geq \D Y_0\geq Y^{\D c}_0 +  \eta_0,$$
 and thus \eqref{6181} holds, which,  via  \eqref{transfY} and \cite[Proposition 2.2]{Hao17}, implies \eqref{6172}, where $\bar \eta = \frac{\eta_0}{\gamma -1}$.
\end{proof}
We will need the following technical lemma.
\begin{Lemma}
\label{lemck}Let $x>0$ be fixed. 
Under the conditions of Theorem \ref{mainThm2}, let $\e_k$, $k\in\bN$, be a sequence or real numbers converging to zero. 
Let us set 
\be\label{defck}
c^{k, \delta', M}_t:= \frac {x}{x+\delta'}\delta'\vee\widehat c_t(x,0)\frac 1{N_t^{\e_k}}\wedge {M}, \quad t\in[0,T],\quad k\in\bN, \quad \delta'>0,\quad M>0.
\ee
Then, for every $k\in\bN$, there exist $\delta'(k)$, $M(k)$, such that for 
\be\label{8111}
\widetilde c^k := c^{k, \delta'(k), M(k)} , \quad k\in\bN,\ee 
the associated solutions to \eqref{BSDEY} satisfy 
$$\lim\limits_{k\to\infty}Y^{\widetilde c^k} = Y^{\widehat c(x,0)},\quad {ucp}.$$
\end{Lemma}
\begin{proof}
First, we observe that  ${\widetilde c^k}\in\cA(x, \e_k)$, $k\in\bN$.   
 Fixing an $\e'>0$, {utilizing the argument from the proof of \cite[Proposition 2.2]{Hao17}}, one can first show that there exists $n'\in\bN$, such that \be\label{817}
\bP\[\sup\limits_{t\in[0,T]}|Y^n_t - Y^{\widehat c(x,0)}_t|\geq \frac{\e'}{3}\]<\frac {\e'}3,\quad n\geq n',\ee 
where $Y^n$ is the first component of the solution to 
$$Y^n_t = (e^{-\delta \theta T} \(\widehat c_T(x,0)\)^{1-\gamma} )\wedge n + \int_t^T F(s, \widehat c_s(x,0), Y^n_s) ds - \int_t^T Z^n_s dB_s,~t\in[0,T].$$
Further {following the proof of \cite[Proposition 2.2]{Hao17}}, one can show that  $(Y^n, Z^n)\in\cS^\infty\times \cM^2$ to the BSDE above exists and is unique. Furthermore, via the comparison result, see, e.g., \cite[Theorem 2.4]{Pard}, we deduce that $0\leq Y^n \leq n$, and $Y^n = \downarrow \lim\limits_{m\to\infty} Y^{n, m}$, $m\in\bN$, where $Y^{n,m}$ solves
\be\label{813}
\bs
Y^{n,m}_t = (e^{-\delta \theta T} \(\widehat c_T(x,0)\)^{1-\gamma} )\wedge n &+ \int_t^T F^m(s, \widehat c_s(x,0), Y^{n, m}_s)ds  \\
&- \int_t^T Z^{n,m}_s dB_s,\quad  t\in[0,T],
\end{split}
\ee where $F^m(t, c_t, y)  : = \delta \theta e^{-\delta t} (c_t^{1- \frac 1 \psi}\wedge m)(|y| \wedge m)^{1- \frac 1 \theta}.$
Here, by comparison, we have  $0\leq Y^{n, m}\leq n$. As a result, for $m\geq n$, we get 
\be\label{812}\bs
F^m(t, c_t, Y^{n, m}_t)   &= \delta \theta e^{-\delta t} (c_t^{1- \frac 1 \psi}\wedge m)(Y^{n, m}_t \wedge m)^{1- \frac 1 \theta} \\
&= \delta \theta e^{-\delta t} (c_t^{1- \frac 1 \psi}\wedge m)(Y^{n, m}_t)^{1- \frac 1 \theta}.
\end{split}
\ee
Further, as we can show that, for every $n\in\bN$, we have $\lim\limits_{m\to\infty} \sup\limits_{t\in[0,T]}|Y_t^n - Y_t^{n, m}| = 0$ in probability $\bP$. Therefore, we conclude that there exists $m'(n)$, such that 
\be\label{811} \bP \[\sup\limits_{t\in[0,T]}|Y_t^n - Y_t^{n, m}|\geq \frac {\e'}{3}\]\leq \frac {\e'}{3},\quad m\geq m'(n).\ee

Next, for  $c^{k, \delta', M}$ given by \eqref{defck}, $k\in\bN$, $\delta'>0$, and $M>0$, 
let us consider the following family of BSDEs 
\be\label{815}\bs
\bar Y^{c^{k, \delta', M}}_t =&  e^{-\delta \theta T} (  c^{k, \delta', M}_T)^{1-\gamma}   +\int_t^T F\(s,   c^{k, \delta', M}_s, \bar Y^{c^{k, \delta', M}}_s\)ds \\
&- \int_t^T \bar Z^{k, \delta', M}_s dB_s,\quad t\in[0,T],\quad k\in\bN,\quad \delta'>0,\quad M>0.
\end{split}
\ee
By 
\cite[Theorem 5.1]{CohenElliot2}, for every choice of ${k, \delta', M}$, there exists a unique solution to \eqref{815},  $\(\bar Y^{c^{k, \delta', M}}, \bar Z^{k, \delta', M} \)\in \cS^2\times \cM^2$. Further, by replacing $F$ with $F^k$ as in \eqref{812}, and using the comparison for BSDEs 
results as in \cite[Theorem 2.4]{Pard}, we deduce that the first component of the solution is in $\cS^\infty$.

Let us consider \eqref{813} and \eqref{815}. These are BSDEs with bounded terminal conditions and Lipschitz generators. Therefore,  for a given $n$ satisfying \eqref{817} and $m$ satisfying \eqref{813},  \cite[Theorem 19.1.6, p. 472]{ElliotCohen}) allows to pick  $\delta'$ and $M$ and then $k_0$, such that \be\label{816}
 \bP \[\sup\limits_{t\in[0,T]}|\bar Y_t^{c^{k, \delta', M}} - Y_t^{n, m}|\geq \frac {\e'}{3}\]<\frac {\e'}3, \quad k\geq k_0.
\ee
Comparing \eqref{817}, \eqref{811}, and \eqref{816}, we deduce that 
$$ \bP \[\sup\limits_{t\in[0,T]}|\bar Y_t^{c^{k, \delta', M}} - Y_t^{\widehat c(x,0)}|\geq   {\e'} \]<  {\e'} , \quad k\geq k_0.$$
As $\e'$ is arbitrary,  we deduce that there exists $\widetilde c^k$, $k\in\bN$, as in \eqref{8111}, for which the assertion of this lemma  holds. 
\end{proof} 
 
\subsection{Proofs of the main theorems}
\begin{proof}[Proof of Theorem \ref{mainThm1}]
First, we observe that \eqref{finValueuv} follows from Lemma \ref{lemFinValue}. Next, from Lemma \ref{lemusc}, we get
\be\label{72312}
\liminf\limits_{(x',\e)\to (x,0)}u(x',\e)\geq u(x,0).
 \ee
 Applying Lemma \ref{lemvsc}, we obtain 
\be\label{72314}
v(y, 0) + xy \geq \limsup\limits_{(y',\e)\to (y, 0)} (v(y',\e) + xy').
\ee
Now, 
using \cite[Theorem 2.7]{MX}, we have 
\be\label{72311}
\limsup\limits_{(y',\e)\to (y, 0)} (v(y', \e) + x'y') \geq \limsup\limits_{\e\to 0}u(x',\e),\quad x'>0.
\ee
From the assumption of the theorem (equation \eqref{conjugacy}), we deduce  
\be\label{72313}
u(x,0) = v(y, 0) + xy.
\ee
Combining \eqref{72312}, \eqref{72314}, \eqref{72311}, and \eqref{72313},  we conclude 
\be\nn\bs
&u(x,0) \leq \liminf\limits_{(x',\e)\to (x,0)}u(x', \e)\leq \limsup\limits_{(x',\e)\to (x,0)}u(x', \e)\\
&\leq \limsup\limits_{(x',\e)\to (x,0)}\(v(y, \e) + x'y\) \leq v(y,0) + xy = u(x,0).
\end{split}
\ee
Therefore, all inequalities above are equalities, and we get  
\eqref{convergenceu}.
Next, from \eqref{72313}, \eqref{72314}, \eqref{72311}, and  \eqref{72312}, we obtain 
\be\nn\bs
&u(x,0) = v(y, 0) + xy \geq \limsup\limits_{(y',\e) \to (y,0)}
\(v(y', \e) + xy'\) \\
&\geq \liminf\limits_{(y',\e) \to (y,0)}
\(v(y', \e) + xy'\)\geq \liminf\limits_{(y',\e) \to (y,0)} u(x,\e)\geq u(x,0).
\end{split}
\ee
Therefore, all inequalities above are actually equalities. This implies  
\eqref{convergencev}.

Finally, the existence and uniqueness of the optimizers 
follows from Lemma \ref{lemConv} and convexity and closedness in $\bL^0(d\kappa\times\bP)$ of the set ${\widetilde\cA}(x,\e)$, $(x,\e)\in(0,\infty)\times(-\e_0,\e_0)$, note that 
the convexity of ${\widetilde\cA}(x,\e)$, $(x,\e)\in(0,\infty)\times(-\e_0,\e_0)$, follows from Assumption \ref{asPert} and Lemma \ref{lemConvA}. 

\end{proof}

\begin{proof}[Proof of Theorem \ref{mainThm2}] {\it Step 1.}
Assume by contradiction that the assertion of this theorem, that is \eqref{convOpt},  fails. Then, there exists $\delta>0$, such that 
$$\limsup\limits_{n\to\infty} (d\kappa\times\bP)\[|\widehat c(x^n, \e^n) - \widehat c(x,0) |>\delta\]>\delta.$$
As $\frac 1{N^{\e_n}}$, $n\in\bN$, converges to $1$, in measure $(d\kappa\times \bP)$, consequently  $\frac 1{N^{\e_n}}$, $n\in\bN$, is bounded in $\bL^0(d\kappa\times \bP)$. Next, following  \eqref{noArb} and the argument in \cite[Proposition 4.2]{Mostovyi2015}, one can see that  
 the set $\cA(1,0)$ is bounded in $\bL^0(d\kappa\times\bP)$. Therefore, since $\widehat c(x^n, \e^n)\frac 1{N^{\e^n}}\in\cA(x^n, 0)$, by possibly passing to smaller $\delta$, we deduce that 
$$\limsup\limits_{n\to\infty} (d\kappa\times\bP)\[ \left|\widehat c(x^n, \e^n)- \widehat c(x,0)\frac 1{N^{\e^n}} \right|\geq\delta, ~~\widehat c(x^n, \e^n) + \widehat c(x,0)\frac 1{N^{\e^n}}\leq\frac 1\delta \]\geq\delta.$$
With $\widetilde c^k$, $k\in\bN$, as in \eqref{8111} (in Lemma \ref{lemck}), by passing to even smaller $\delta$, we get
\be\label{8141}
\limsup\limits_{n\to\infty} (d\kappa\times\bP)\[\left|\widehat c(x^n, \e^n)-\widetilde  c^n \right|\geq\delta, ~~\widehat c(x^n, \e^n) + \widetilde c^n\leq\frac 1\delta\]\geq\delta.
\ee
Let us set \be\label{8152}
\bar c^n:=\frac 12\left|\widehat c(x^n, \e^n) +\widetilde c^n\right|\in\cA\(\frac {x^n + x}2, \e^n\),\quad n\in\bN.
\ee
Furthermore, one can show that $\bar c^n\in{\widetilde\cA}\(\frac {x^n + x}2, \e^n\)$, as for every $t\in[0,T]$, we have 
\be\nn\bs\(\bar c^n_t\)^{1-\frac 1\psi} &=\(\frac 12\left|\widehat c_t(x^n, \e^n) +\widetilde c^n_t\right|\)^{1 - \frac 1\psi}\\ &\leq \max\( \widehat c_t(x^n, \e^n),   \widetilde c^n_t\)^{1 - \frac 1\psi}\leq \(\widehat c_t(x^n, \e^n)\)^{1 - \frac 1\psi} + (\widetilde c^n_t)^{1 - \frac 1\psi},
\end{split}\ee
and at maturity, we have 
$$\(\bar c^n_T\)^{1-\gamma}= \(\frac 12\left|\widehat c_T(x^n, \e^n) +\widetilde c^n_T\right|\)^{1-\gamma} \leq \(\frac 12\frac{x}{x + \e^n}\frac 1{\e^n} \)^{1-\gamma},$$
and thus, $\bar c^n$'s satisfy both  integrability conditions in  
the definition of ${\widetilde\cA}$'s in \eqref{defAa}. 

{\it Step 2.}
 Let us use Lemma \ref{lemConv} along a subsequence from {\it Step 1} that we do not relabel and such that
 \be\nn\label{81412}
\lim\limits_{n\to\infty} (d\kappa\times\bP)\[\left|\widehat c(x^n, \e^n)-\widetilde  c^n \right|\geq\delta, ~~\widehat c(x^n, \e^n) + \widetilde c^n\leq\frac 1\delta\]\geq\delta.
\ee
  In the argument below,  the notations from the proof of Lemma \ref{lemConv} are used.  
 For 
 \be\label{8154}
 \eta^n := \bE\[(\Gamma^n\xi^n)\cdot \kappa_T \],\quad n\in\bN,
 \ee one can show that\footnote{{The lower bounds on $\alpha^n$'s are obtained through estimates on the slopes of $F$.}} 
 $$\Gamma^n_t \geq \exp\(a\theta\int_0^t \(\bar c^n_s\)^{1-\frac 1\psi}\(\D Y^{n}_s\)^{-\frac 1\theta}ds\),
 $$
 for some constant $a>0$ (and where $\theta<0$ and {$\D Y^{n}$ is as in the proof of Lemma \ref{lemConv} corresponding to $c'=\widehat c(x^n, \e^n)$,  $c''=\widetilde c^n$, and $\lambda =\frac 12$.})

 Next, from Lemma \ref{lemck},  along a subsequence, which we do not relabel, we have 
\be\nn\label{8151}
\lim\limits_{k\to\infty}\sup\limits_{t\in[0,T]}\left|Y^{\widetilde c^k}_t - Y^{\widehat c(x,0)}_t\right| = 0, \quad \Pas.
\ee
 Further,   $Y^{\widehat c(x^n, \e^n)}$ is bounded from above by a real-valued process,   by the assumption of this theorem.
Therefore, as 
 $\D Y^{n}_t\leq \frac 12 Y_t^{\widehat c(x^n, \e^n)}+ \frac 12Y_t^{\widetilde c^n}$, $t\in[0,T]$, $\Pas$,  by the proof of Lemma \ref{lemConv}, we deduce that 
 \be\bs
 \Gamma^n_t 
 &\geq \exp\(a\theta\int_0^t (\bar c^n_s)^{1-\frac 1\psi}(\D Y^{n}_s)^{-\frac 1\theta}ds\)\\
 &\geq
 \exp\(a\theta\int_0^t \left(\widehat c_s(x^n, \e^n) +\widetilde c_s^n\right)^{1-\frac 1\psi}( Y_s^{\widehat c(x^n, \e^n)}+  Y_s^{\widetilde c^n})^{-\frac 1\theta}ds\)
.
  \end{split}
 \ee
 From the assumptions of this theorem 
  and Lemma \ref{lemck}, we obtain that 
\be\label{821}\bs
&\liminf\limits_{n\to\infty}\Gamma^n_t 
=:\widetilde \Gamma^\infty_t>0,\quad t\in[0,T],\quad \Pas.
\end{split}
\ee
 Let us consider the sequence $\xi^n$, $n\in\bN$. Following the proof of Lemma \ref{lemConv} (see \eqref{8141}), we observe that 
\be\label{8147}\xi^n_t\geq \widetilde \delta_1(\Delta Y_t^n)^{1-\frac 1\theta}1_{\{\left|\widehat c_t(x^n, \e^n)-\widetilde  c^n_t \right|\geq\delta, ~~\widehat c_t(x^n, \e^n) + \widetilde c^n_t\leq\frac 1\delta\}}, \quad t\in[0,T), \ee
and
\be\label{8146} \xi^n_T \geq\delta_21_{\{\left|\widehat c_T(x^n, \e^n)-\widetilde  c^n_T \right|\geq\delta, ~~\widehat c_T(x^n, \e^n) + \widetilde c^n_T\leq\frac 1\delta\}},
\ee
where constant $\widetilde \delta_1>0$ and $\delta_2>0$ depend on $\delta$ appearing in \eqref{8141} only. 
As $\D Y^n_t\geq \frac 1{2^\theta} Y^{\widetilde c^n}_t$, $t\in[0,T]$, $\Pas$, by the argument in Lemma \ref{lemck}
, we have 
\be\label{8145}
\liminf\limits_{n\to\infty}\(\D Y^n_t \)\geq\liminf\limits_{n\to\infty} \frac 1{2^\theta}Y^{\widetilde c^n}_t= \frac 1{2^\theta} Y_t^{\widehat c(x,0)}>0,~~ t\in[0,T],~~ \Pas.
\ee
  
 By the Dunford-Pettis compactness criterion (see, e.g.,  \cite[p. 26]{KaratzasShreve1}), there exists a weakly (in $\bL^1(d\kappa\times\bP)$) convergent subsequence of $\xi^n\wedge 1$, $n\in\bN$, whose limit is denoted by $\xi^\infty$. In view of 
 \eqref{8147}, \eqref{8146}, and \eqref{8145}, we have  $(d\kappa\times \bP)[\xi^\infty>0]>0$. 

{Let us pass to this subsequence that we do not relabel again}. The non-negativity of $\Gamma^n\xi^n$ (by the construction above) allows invoking  Fatou's lemma, which implies  that 
\be\nn\bs
&\liminf\limits_{n\to\infty}\bE\[(\Gamma^n(\xi^n\wedge 1))\cdot \kappa_T\] \\
=&\liminf\limits_{n\to\infty}
\( 
\bE\[(\widetilde \Gamma^\infty(\xi^n\wedge 1))\cdot \kappa_T\] + 
\bE\[\((\Gamma^n - \widetilde \Gamma^\infty)(\xi^n\wedge 1)\)\cdot \kappa_T\] 
\) \\
\geq& \bE\[(\widetilde\Gamma^\infty \xi^\infty)\cdot \kappa_T\]>0,
\end{split}
\ee
as  $(d\kappa\times \bP)[\xi^\infty>0]>0$ and $\widetilde \Gamma^\infty>0$, $(d\kappa\times \bP)-a.e.$,  as well as 
$$\lim\limits_{n\to\infty}\bE\[(\widetilde \Gamma^\infty(\xi^n\wedge 1))\cdot \kappa_T\] = \bE\[(\widetilde \Gamma^\infty\xi^\infty)\cdot \kappa_T\],$$ by the weak convergence in $\bL^1(d\kappa\times\bP)$ (here we recall that $\Gamma$'s take values in $[0,1]$ as $\theta<0$),  by Fatou's lemma (here, $\Gamma^\infty(\xi^n\wedge 1)$ is bounded from below by $-1$)  and \eqref{821}, we have
$$\liminf\limits_{n\to\infty}\bE\[\((\Gamma^n - \widetilde \Gamma^\infty)(\xi^n\wedge 1)\)\cdot \kappa_T\] \geq 
0.$$  
We conclude that 
\be\label{7311}
\liminf\limits_{n\to\infty}\eta^n>0.
\ee

{\it Step 3.} For $\bar c^n$, $n\in\bN$, defined in {\it Step 1} (see \eqref{8152}), let us consider the subsequence from {\it Step 2}. 
By Lemma \ref{lemConv}, we have
\be\label{7241}
\liminf_{n\to\infty}U^{\bar c^n}_0 \geq \liminf\limits_{n\to\infty} \( \frac 12u(x^n, \e^n) + \frac 12 U^{c^n}_0+ \bar\eta^n\),
\ee
where $\bar \eta^n = \frac {\eta^n}{\gamma-1}$ and $\eta^n$ are given in \eqref{8154}. 
It follows from Lemma \ref{lemck} that 
\be\label{7242}\lim \limits_{n\to\infty}U^{c^n}_0 = u(x,0).
\ee
On the other hand, as $\bar c^n\in{\widetilde\cA}\(\frac {x^n + x}2, \e^n\)$, we get 
\be\label{7243}
U^{\bar c^n}_0 \leq u\(\frac {x^n + x}2, \e^n\).
\ee
By Theorem \ref{mainThm1}, we have
\be\label{7244}
\liminf\limits_{n\to\infty}  u(x^n, \e^n) = u(x,0).
\ee
Therefore, in \eqref{7241}, via \eqref{7242}, \eqref{7243}, and \eqref{7244}, we conclude that  
$$u(x,0)\geq \liminf_{n\to\infty}U^{\bar c^n}_0 \geq \liminf\limits_{n\to\infty}  \( \frac 12 u(x^n, \e^n) + \frac 12 U^{\widetilde c^n}_0+ \bar \eta^n\)  \geq u(x,0) +\liminf\limits_{n\to\infty} \bar\eta^n,$$
which is impossible, as $\liminf\limits_{n\to\infty}\bar\eta^n = \liminf\limits_{n\to\infty}\frac{\eta^n}{\gamma-1}>0$ by \eqref{7311}.

%
%
\end{proof}
{ \begin{proof}[Proof of Theorem \ref{thmConvU}]
Let us recall that, for a given nonnegative consumption stream $c$,  $U^c$ was defined in \eqref{defUEZ} and $Y^c$ in \eqref{BSDEY}. The proof of Theorem \ref{thmConvU} is entirely similar to the proof of Lemma \ref{lemck}.  It relies on the truncation and the stability of BSDEs result as in \cite[Theorem 19.1.6, p. 472]{ElliotCohen}, so that we can show that  
 $$\lim\limits_{(x',\e) \to (x,0)}Y^{\widehat c(x', \e)} =  Y^{\widehat c(x,0)},\quad ucp\quad {\rm and} \quad 
\lim\limits_{(x',\e) \to (x,0)}U^{\widehat c(x', \e)} =  U^{\widehat c(x,0)},\quad ucp.$$
We omit further details for brevity. 
 \end{proof}
\begin{proof}[Proof of Theorem \ref{lemWealthPr}]
Let us consider a sequence $(x_n,\e_n)$, $n\in\bN$, convergent to $(x,0)$. Without loss of generality, we will suppose that $x_n>0$ and $\e_n\in(-\e_0, \e_0)$, $n\in\bN$. 
Let us denote $$X^n = X^{\widehat c(x_n,\e_n)}, \quad N^n = N^{\e_n}, \quad c^n = \widehat c(x_n, \e_n), \quad n\in\bN,\quad D = \widehat D(y,0),$$  and set 
\be\label{3278}
\bs
D^n&: = DN^n,\quad {L}^n : = \frac 1{x_n}\left(X^nD^n + \int_0^\cdot D^n_s c^n_sd\kappa_s \right).
\end{split}\ee
Then, by Lemma \ref{lemAdm}, $D^n\in\cD( y, \e_n)$ and thus ${L}^n$, $n\in\bN$, is a sequence of nonnegative supermartingales. 
Since $(d\kappa\times\bP) \text-\lim\limits_{n\to\infty}c^n = \widehat c(x, 0)$ by Theorem \ref{mainThm2} and  $(d\kappa\times\bP)\text-\lim\limits_{n\to\infty}D^n = D$ by \eqref{3278} and the assumption that $(d\kappa\times\bP) \text-\lim\limits_{\e\to 0}N^\e =1$, we pass to a subsequence, which we do not relabel and suppose that 
$\lim\limits_{n\to\infty}D^nc^n = D\widehat c(x, 0)$, $(d\kappa\times\bP)$-a.e..  Therefore, using Fatou's lemma, we get 
\be\label{3271}
\liminf\limits_{n\to\infty}\int_0^T D^n_s c^n_sd\kappa_s \geq \int_0^T D_s \widehat c_s(x,0)d\kappa_s, \quad \Pas. 
\ee
Let us further set 
\be\label{defL}
{L}: = \frac 1x\left( X^{\widehat c(x,0)} D + \int_0^\cdot D_s\widehat c_s(x,0)d\kappa_s\right).
\ee
The optimality of $\widehat c(x,0)$ implies that $X^{\widehat c(x,0)}_T = 0$, $\Pas$, as it is optimal to consume everything that is left at maturity.  Likewise, the optimality of $c^n$ implies that $X^n_T = 0$, $\Pas$, for every $n\in\bN$. Therefore, \eqref{3278}, \eqref{3271}, and \eqref{defL} result in  
\be\label{3272}
\liminf\limits_{n\to\infty}{L}^n_T \geq {L}_T,\quad \Pas.
\ee
From the respective definitions of $L^n$, $n\in\bN$, and $L$, we conclude that 
\be\label{3275}
{L}^n_0 =  D_0 = {L}_0 , \quad n\in\bN.
\ee
As ${L}^n$, $n\in\bN$, are nonnegative c\`adl\`ag supermartingales and ${L}$ is a nonnegative c\`adl\`ag  martingale, from \eqref{3272} and \eqref{3275}, one can show that 
\be\label{3276}
\lim\limits_{n\to\infty}{L}^n = {L},\quad ucp.
\ee
Consequently, and in particular, passing to another subsequence, which we do not relabel, we get
\be\label{32711}
\lim\limits_{n\to\infty}\int_0^T D^n_s c^n_sd\kappa_s = \int_0^T D_s \widehat c_s(x,0)d\kappa_s, \quad \Pas. 
\ee
Next, similarly to \eqref{3271}, for every $t\in[0,T]$, we deduce that  
\be\label{481}
\liminf\limits_{n\to\infty}\int_{t}^{T} D^n_s c^n_sd\kappa_s \geq \int_{t}^{T} D_s \widehat c_s(x,0)d\kappa_s, \quad \Pas, 
\ee
and 
\be\label{482}
\liminf\limits_{n\to\infty}\int_{0}^{t} D^n_s c^n_sd\kappa_s \geq \int_{0}^{t} D_s \widehat c_s(x,0)d\kappa_s, \quad \Pas.
\ee
Therefore, from \eqref{32711} and \eqref{481}, we get 
\be\label{483}
\bs
\limsup\limits_{n\to\infty}\int_{0}^{t} D^n_s c^n_sd\kappa_s 
&= \limsup\limits_{n\to\infty}\(\int_{0}^{T} D^n_s c^n_sd\kappa_s - \int_{t}^{T} D^n_s c^n_sd\kappa_s\)\\
&\leq \int_{0}^{t} D_s \widehat c_s(x,0)d\kappa_s, \quad \Pas. 
\end{split}
\ee
In turn, \eqref{482} and \eqref{483} imply that 
\be\label{484}
\lim\limits_{n\to\infty}\int_{0}^{t} D^n_s c^n_sd\kappa_s = \int_{0}^{t} D_s \widehat c_s(x,0)d\kappa_s, \quad \Pas,
\ee
where the last equality holds for every $t\in[0,T]$. 
As the processes $\int_{0}^{\cdot} D^n_s c^n_sd\kappa_s$, $n\in\bN$, and $\int_{0}^{\cdot} D_s \widehat c_s(x,0)d\kappa_s$, are c\'adl\'ag monotone, from \eqref{484}, we get 
\be\label{3279}
\lim\limits_{n\to\infty}\sup\limits_{t\in[0,T]}\left|\int_0^t D^n_s c^n_sd\kappa_s - \int_0^t D_s \widehat c_s(x,0)d\kappa_s\right| = 0, \quad \Pas.
\ee 
Finally, as $D^n$'s and $D$ are strictly positive and $(d\kappa\times\bP)\text-\lim\limits_{n\to\infty} D^n = D$, from \eqref{3276} and \eqref{3279}, using \cite[Thoerem 1.6.2, p. 46]{Du}, we deduce  that $X^n = \frac {{L}^n - \int_0^\cdot D^n_sc^n_sd\kappa_s}{D^n}$, $n\in\bN$, converges to $X^{\widehat c(x,0)}=\frac {{L} - \int_0^\cdot D_s\widehat c_s(x,0)d\kappa_s}{D}$ in measure $(d\kappa\times\bP)$. If $\lim\limits_{\e\to 0 }N^\e= 1$, $ucp$, then, similarly, from \eqref{3278},  \eqref{3276}, and \eqref{3279}, using \cite[Thoerem 1.6.2, p. 46]{Du}, we conclude that $$\lim\limits_{n\to\infty} X^n = X^{\widehat c(x,0)},\quad ucp.$$
\end{proof}
}
\bibliographystyle{alpha}\bibliography{literature1}
\end{document}